\documentclass[a4paper,conference]{IEEEtran}
\usepackage{graphicx,ifthen}
\usepackage{psfrag,amssymb,amsthm}
\usepackage[cmex10]{amsmath}
\usepackage{cite,flushend,color,pst-plot,stfloats,pst-sigsys}
\usepackage{pstricks-add}

\newtheorem{thm}{Theorem}
\newtheorem{lem}{Lemma}
\newtheorem{cor}{Corollary}
\newtheorem{con}{Conjecture}
\newtheorem{prop}{Proposition}

\theoremstyle{definition}
\newtheorem{definition}{Definition}

\def \arxiv {1}

\title{Relative Information Loss in the PCA}

\author{\IEEEauthorblockN{Bernhard C. Geiger\IEEEauthorrefmark{1}, Gernot Kubin\IEEEauthorrefmark{1}
\IEEEauthorblockA{\IEEEauthorrefmark{1}Signal Processing and Speech Communication Laboratory, Graz University of Technology, Austria}
$\{$geiger,gernot.kubin$\}$@tugraz.at}}

\begin{document}
\newcounter{myTempCnt}

\ifthenelse{\arxiv=1}{
\newcommand{\x}[1]{x[#1]}
\newcommand{\y}[1]{y[#1]}

\newcommand{\pdfy}{f_Y(y)}

\newcommand{\ent}[1]{H(#1)}
\newcommand{\diffent}[1]{h(#1)}
\newcommand{\derate}[1]{\bar{h}\left(\mathbf{#1}\right)}
\newcommand{\mutinf}[1]{I(#1)}
\newcommand{\ginf}[1]{I_G(#1)}
\newcommand{\kld}[2]{D(#1||#2)}
\newcommand{\kldrate}[2]{\bar{D}(\mathbf{#1}||\mathbf{#2})}
\newcommand{\binent}[1]{H_2(#1)}
\newcommand{\binentneg}[1]{H_2^{-1}\left(#1\right)}
\newcommand{\entrate}[1]{\bar{H}(\mathbf{#1})}
\newcommand{\mutrate}[1]{\mutinf{\mathbf{#1}}}
\newcommand{\redrate}[1]{\bar{R}(\mathbf{#1})}
\newcommand{\pinrate}[1]{\vec{I}(\mathbf{#1})}
\newcommand{\loss}[2][\empty]{\ifthenelse{\equal{#1}{\empty}}{L(#2)}{L_{#1}(#2)}}
\newcommand{\lossrate}[2][\empty]{\ifthenelse{\equal{#1}{\empty}}{L(\mathbf{#2})}{L_{\mathbf{#1}}(\mathbf{#2})}}
\newcommand{\gain}[1]{G(#1)}
\newcommand{\atten}[1]{A(#1)}
\newcommand{\relLoss}[2][\empty]{\ifthenelse{\equal{#1}{\empty}}{l(#2)}{l_{#1}(#2)}}
\newcommand{\relLossrate}[1]{l(\mathbf{#1})}
\newcommand{\relTrans}[1]{t(#1)}
\newcommand{\partEnt}[2]{H^{#1}(#2)}

\newcommand{\dom}[1]{\mathcal{#1}}
\newcommand{\indset}[1]{\mathbb{I}\left({#1}\right)}

\newcommand{\unif}[2]{\mathcal{U}\left(#1,#2\right)}
\newcommand{\chis}[1]{\chi^2\left(#1\right)}
\newcommand{\chir}[1]{\chi\left(#1\right)}
\newcommand{\normdist}[2]{\mathcal{N}\left(#1,#2\right)}
\newcommand{\Prob}[1]{\mathrm{Pr}(#1)}
\newcommand{\Mar}[1]{\mathrm{Mar}(#1)}
\newcommand{\Qfunc}[1]{Q\left(#1\right)}

\newcommand{\expec}[1]{\mathrm{E}\left\{#1\right\}}
\newcommand{\expecwrt}[2]{\mathrm{E}_{#1}\left\{#2\right\}}
\newcommand{\var}[1]{\mathrm{Var}\left\{#1\right\}}
\renewcommand{\det}{\mathrm{det}}
\newcommand{\cov}[1]{\mathrm{Cov}\left\{#1\right\}}
\newcommand{\sgn}[1]{\mathrm{sgn}\left(#1\right)}
\newcommand{\sinc}[1]{\mathrm{sinc}\left(#1\right)}
\newcommand{\e}[1]{\mathrm{e}^{#1}}
\newcommand{\multint}{\iint{\cdots}\int}
\newcommand{\modd}[3]{((#1))_{#2}^{#3}}
\newcommand{\quant}[1]{Q\left(#1\right)}
\newcommand{\card}[1]{\mathrm{card}(#1)}
\newcommand{\diam}[1]{\mathrm{diam}(#1)}

\newcommand{\ivec}{\mathbf{i}}
\newcommand{\hvec}{\mathbf{h}}
\newcommand{\gvec}{\mathbf{g}}
\newcommand{\avec}{\mathbf{a}}
\newcommand{\kvec}{\mathbf{k}}
\newcommand{\fvec}{\mathbf{f}}
\newcommand{\vvec}{\mathbf{v}}
\newcommand{\xvec}{\mathbf{x}}
\newcommand{\Xvec}{\mathbf{X}}
\newcommand{\Xhvec}{\hat{\mathbf{X}}}
\newcommand{\xhvec}{\hat{\mathbf{x}}}
\newcommand{\xtvec}{\tilde{\mathbf{x}}}
\newcommand{\Yvec}{\mathbf{Y}}
\newcommand{\yvec}{\mathbf{y}}
\newcommand{\Zvec}{\mathbf{Z}}
\newcommand{\Svec}{\mathbf{S}}
\newcommand{\Nvec}{\mathbf{N}}
\newcommand{\Pvec}{\mathbf{P}}
\newcommand{\muvec}{\boldsymbol{\mu}}
\newcommand{\wvec}{\mathbf{w}}
\newcommand{\Wvec}{\mathbf{W}}
\newcommand{\Hmat}{\mathbf{H}}
\newcommand{\Amat}{\mathbf{A}}
\newcommand{\Fmat}{\mathbf{F}}

\newcommand{\zerovec}{\mathbf{0}}
\newcommand{\eye}{\mathbf{I}}
\newcommand{\evec}{\mathbf{i}}

\newcommand{\zeroone}{\left[\begin{array}{c}\zerovec^T\\ \eye\end{array} \right]}
\newcommand{\zerooneT}{\left[\begin{array}{cc}\zerovec & \eye\end{array} \right]}
\newcommand{\zerooneM}{\left[\begin{array}{cc}\zerovec &\zerovec^T\\\zerovec& \eye\end{array} \right]}

\newcommand{\Cxx}{\mathbf{C}_{XX}}
\newcommand{\Cx}{\mathbf{C}_{\Xvec}}
\newcommand{\Chx}{\hat{\mathbf{C}}_{\Xvec}}
\newcommand{\Cy}{\mathbf{C}_{\Yvec}}
\newcommand{\Cz}{\mathbf{C}_{\Zvec}}
\newcommand{\Cn}{\mathbf{C}_{\mathbf{N}}}
\newcommand{\Cnt}{\underline{\mathbf{C}}_{\tilde{\mathbf{N}}}}
\newcommand{\Cntm}{\underline{\mathbf{C}}_{\tilde{\mathbf{N}}}}
\newcommand{\Cxh}{\mathbf{C}_{\hat{X}\hat{X}}}
\newcommand{\rxx}{\mathbf{r}_{XX}}
\newcommand{\Cxy}{\mathbf{C}_{XY}}
\newcommand{\Cyy}{\mathbf{C}_{YY}}
\newcommand{\Cnn}{\mathbf{C}_{NN}}
\newcommand{\Cyx}{\mathbf{C}_{YX}}
\newcommand{\Cygx}{\mathbf{C}_{Y|X}}
\newcommand{\Wmat}{\underline{\mathbf{W}}}

\newcommand{\Jac}[2]{\mathcal{J}_{#1}(#2)}

\newcommand{\NN}{{N{\times}N}}
\newcommand{\perr}{P_e}
\newcommand{\perh}{\hat{\perr}}
\newcommand{\pert}{\tilde{\perr}}

\newcommand{\vecind}[1]{#1_0^n}
\newcommand{\roots}[2]{{#1}_{#2}^{(i_{#2})}}
\newcommand{\rootx}[1]{x_{#1}^{(i)}}
\newcommand{\rootn}[2]{x_{#1}^{#2,(i)}}

\newcommand{\markkern}[1]{f_M(#1)}
\newcommand{\pole}{a_1}
\newcommand{\preim}[1]{g^{-1}[#1]}
\newcommand{\preimV}[1]{\mathbf{g}^{-1}[#1]}
\newcommand{\Xmax}{\bar{X}}
\newcommand{\Xmin}{\underbar{X}}
\newcommand{\xmax}{x_{\max}}
\newcommand{\xmin}{x_{\min}}
\newcommand{\limn}{\lim_{n\to\infty}}
\newcommand{\limX}{\lim_{\hat{\Xvec}\to\Xvec}}
\newcommand{\limx}{\lim_{\hat{X}\to X}}
\newcommand{\limXo}{\lim_{\hat{X}_1\to X_1}}
\newcommand{\sumin}{\sum_{i=1}^n}
\newcommand{\finv}{f_\mathrm{inv}}
\newcommand{\ejtheta}{\e{\jmath\theta}}
\newcommand{\khat}{\bar{k}}
\newcommand{\modeq}[1]{g(#1)}
\newcommand{\partit}[1]{\mathcal{P}_{#1}}
\newcommand{\psd}[1]{S_{#1}(\e{\jmath \theta})}
\newcommand{\borel}[1]{\mathfrak{B}(#1)}
\newcommand{\infodim}[1]{d(#1)}

\newcommand{\delay}[2]{\psblock(#1){#2}{\footnotesize$z^{-1}$}}
\newcommand{\Quant}[2]{\psblock(#1){#2}{\footnotesize$\quant{\cdot}$}}
\newcommand{\moddev}[2]{\psblock(#1){#2}{\footnotesize$\modeq{\cdot}$}}}{}

\newcommand{\Ymat}{\underline{\Yvec}}\newcommand{\ymat}{\underline{\yvec}}
\newcommand{\Xmat}{\underline{\Xvec}}
\renewcommand{\Cx}{\underline{\mathbf{C}}_{\Xvec}}
\renewcommand{\Chx}{\underline{\hat{\mathbf{C}}}_{\Xvec}}
\renewcommand{\Cy}{\underline{\mathbf{C}}_{\Yvec}}

\maketitle

\begin{abstract}
In this work we analyze principle component analysis (PCA) as a deterministic input-output system. We show that the relative information loss induced by reducing the dimensionality of the data after performing the PCA is the same as in dimensionality reduction without PCA. Furthermore, we analyze the case where the PCA uses the sample covariance matrix to compute the rotation. If the rotation matrix is not available at the output, we show that an infinite amount of information is lost. The relative information loss is shown to decrease with increasing sample size.
\end{abstract}

\section{Introduction}\label{sec:intro}
Principle component analysis (PCA) is a powerful tool for both linear decorrelation and dimensionality reduction, and is thus widely used in machine learning, neural networks, and feature extraction~\cite{Deco_ITNN,Principe_ITLearning}. A vast literature proves the optimality of PCA in an information theoretic sense for specific cases: Linsker~\cite{Linsker_InfoMax} proved that for \mbox{$N$-dimensional} Gaussian data corrupted by Gaussian noise with diagonal covariance matrix, the PCA maximizes the mutual information between the data and a one-dimensional output random variable. Plumbley argued in~\cite{Plumbley_TN} that for non-Gaussian data the PCA is the linear transform which minimizes an upper bound on the information lost due to dimensionality reduction. Furthermore, the authors of~\cite{Deco_ITNN} show that under some circumstances the PCA minimizes redundancy in the output data, i.e., is an optimal linear independent component analysis. While this list of previous works is clearly not complete, it highlights the utility of information-theoretic measures to characterize the performance of PCA algorithms.

In the present work, we pursue a slightly different approach: We view the PCA as a multivariate, vector-valued input-output system and analyze it in terms of information loss. In this context, \emph{information} means the \emph{total information} available at the input of the system, in contrast to~\cite{Linsker_InfoMax,Plumbley_TN}, which considered only information which is \emph{relevant}. As a consequence, while our analysis may superficially appear to contradict the results available in the literature, it rather provides a different view on PCA as an information-processing system.

In its standard notation, PCA is a linear system
\begin{equation}
 \Yvec = \underline{\mathbf{W}}^T\Xvec
\end{equation}
where $\Xvec$ and $\Yvec$ are the $N$-dimensional input and output vectors, respectively, and $\underline{\mathbf{W}}$ is an orthogonal matrix. In this sense, the PCA is a bijective transform and thus invertible. Typically, however, the output vector will be truncated after performing the PCA in order to reduce the dimensionality of the data. In this projection to a lower-dimensional subspace information is lost, and we intend to quantify the loss in this work.

Let us be more specific: Assuming continuous-valued random variables (RV), the information available at the input is infinite. In particular, assuming joint continuity, each component of the multi-dimensional input RV contains an infinite amount of information; thus, if the dimension is reduced, i.e., if some components are dropped, an infinite amount of information is lost (cf.~\cite{Geiger_ILStatic_IZS}). 

In case the orthogonal matrix is not known a priori but has to be estimated from a set of input data vectors collected in the matrix $\Xmat$, the PCA becomes a nonlinear operation:
\begin{equation}
 \Ymat = \underline{\mathbf{w}}(\Xmat)\Xmat.\label{eq:samplePCA}
\end{equation}
Here, $\underline{\mathbf{w}}$ is a matrix-valued function which computes the orthogonal matrix required for rotating the data (e.g., using the QR algorithm\ifthenelse{\arxiv=1}{~\cite{Golub_MatrixComputations}}{}). If this orthogonal matrix is not stored and made available at the output, one will readily agree that information is lost even if -- at a first glance -- the dimension of the data is not reduced.

Our notion of \emph{relative} information loss captures the ratio of input information which cannot be retrieved at the output. We make this statement precise in Section~\ref{sec:review}. Before analyzing the information loss in PCA in Section~\ref{ssec:PCA}, we present a general theorem about the loss in systems which reduce the dimensionality of the data (Section~\ref{sec:dimredFunct}). Section~\ref{ssec:Toy} acts as a bridge, containing some toy examples which should give an intuitive understanding of relative information loss. In Section~\ref{sec:pcaloss}, the PCA using an input data matrix (cf.~\eqref{eq:samplePCA}) to perform the rotation is shown to destroy information even if the dimensionality is not reduced. Since none of our findings explains the usefulness of the PCA, we finally give an outlook of how to corroborate its optimality using a different notion of information loss in Section~\ref{sec:relevant}. \ifthenelse{\arxiv=1}{We there also discuss possible implications for a system theory from an information-theoretic point-of-view.}{}

\ifthenelse{\arxiv=1}{}{ An extended version of this paper, including some proofs and additional examples, can be found in~\cite{Geiger_RILPCA_arXiv}.}


\section{Relative Information Loss -- A Quick Overview}\label{sec:review}
In this section we will briefly present the basic properties of relative information loss. To this end, we introduce
\begin{definition}[Relative Information Loss]\label{def:relloss}
 Let $\Xvec$ be an $N$-dimensional RV valued in $\dom{X}$, and let $\Yvec$ be obtained by transforming $\Xvec$ with a static function $\gvec$, i.e., $\Yvec=\gvec(\Xvec)$. We define the \emph{relative information loss} induced by this transform as
\begin{equation}
 \relLoss{\Xvec\to\Yvec} = \limn \frac{\ent{\hat{\Xvec}_n|\Yvec}}{\ent{\hat{\Xvec}_n}}
\end{equation}
where $\hat{\Xvec}_n=\frac{\lfloor n\Xvec\rfloor}{n}$ (element-wise). 
The quantity on the left is defined if the limit on the right-hand side exists.
\end{definition}

Loosely speaking, this quantity represents the percentage of input information which is lost in the system, with every bit weighted equally. Taking, for example, an 8 bit number\ifthenelse{\arxiv=1}{\footnote{I.e., $X$ is a scalar RV which can be represented by eight independent Bernoulli-$\frac{1}{2}$ RVs; thus $\ent{X}=8$ bit.}}{}, losing the most significant or the least significant bit both amounts to a relative information loss of $\relLoss{X\to Y}=\frac{1}{8}$. The advantage of this definition is its independence of application-specific aspects (where, e.g., the most significant bit may be more important than the least significant one; cf. Section~\ref{sec:relevant} for a discussion).


The motivation for introducing this quantity is to complement the absolute notion of information loss, given by $\loss{\Xvec\to\Yvec}=\ent{\Xvec|\Yvec}$, which suffers from being infinite in many practically relevant cases (cf.~\cite{Geiger_ILStatic_IZS}).

Due to the non-negativity of entropy and the fact that conditioning reduces entropy, it follows that \mbox{$\relLoss{\Xvec\to\Yvec}\in[0,1]$}. Moreover, for continuous input RVs $\Xvec$, a nonzero relative loss corresponds to an infinite absolute loss:
\begin{prop}\label{prop:infLoss}
 Let $\Xvec$ be such that $\ent{\Xvec}=\infty$ and let\\ $\relLoss{\Xvec\to\Yvec}>0$. Then, $\loss{\Xvec\to\Yvec}=\ent{\Xvec|\Yvec}=\infty$.
\end{prop}
\begin{IEEEproof}
We prove the proposition by contradiction. To this end, assume that $\ent{\Xvec|\Yvec}=L\leq\infty$. Then,
\begin{IEEEeqnarray}{RCL}
 \relLoss{\Xvec\to\Yvec} &=& \limn \frac{\ent{\hat{\Xvec}_n|\Yvec}}{\ent{\hat{\Xvec}_n}} =\limn\inf \frac{\ent{\hat{\Xvec}_n|\Yvec}}{\ent{\hat{\Xvec}_n}}\\
&\leq& \limn\inf\frac{\ent{\Xvec|\Yvec}}{\ent{\hat{\Xvec}_n}}\\ &=& \limn\inf\frac{L}{\ent{\hat{\Xvec}_n}}=0
\end{IEEEeqnarray}
where the inequality is due to data processing. The last equality follows from the fact that at least a subsequence of $\ent{\hat{\Xvec}_n}$ converges to $\ent{\Xvec}$ (cf.~\cite{Pinsker_InfoEngl}).
\end{IEEEproof}

Furthermore, we can show a tight connection to the information dimension introduced by R\'{e}nyi~\cite{Renyi_InfoDim}. We therefore need the following

\begin{definition}[Information Dimension~\cite{Renyi_InfoDim}]\label{def:infodim}
The \emph{information dimension} of an RV $\Xvec$ is given as
\begin{equation}
 \infodim{\Xvec} = \limn \frac{\ent{\hat{\Xvec}_n}}{\log n}
\end{equation}
provided the limit on the right exists.
\end{definition}

R\'{e}nyi then showed that the asymptotic behavior of $\ent{\hat{\Xvec}_n}$ depends strongly on the information dimension of $\Xvec$:

\begin{lem}[Asymptotic behavior of $\ent{\hat{\Xvec}_n}$]\label{lem:epsilon}
Let $\Xvec$ be an RV with existing information dimension $\infodim{\Xvec}$ and let $\ent{\hat{\Xvec}_1}<\infty$. Then, for $n\to\infty$ the entropy of $\hat{\Xvec}_n$ behaves as
\begin{equation}
  \ent{\hat{\Xvec}_n} = \infodim{\Xvec}\log n + h + o(1)
\end{equation}
where $h$ is the $\infodim{\Xvec}$-dimensional entropy of $\Xvec$ (provided it exists).
\end{lem}
\begin{proof}
 See~\cite{Renyi_InfoDim} (cf. also~\cite{Kolmogorov_ContinuousRVs,Kolmogorov_EpsilonEntropy}).
\end{proof}

In particular, if $\Xvec$ has a probability measure absolutely continuous w.r.t. the $N$-dimensional Lebesgue measure $\mu^N$ ($P_\Xvec\ll\mu^N$), then the information dimension $\infodim{\Xvec}=N$ and $h$ denotes the differential entropy of $\Xvec$ (see~Theorems~1 \& 4 in~\cite{Renyi_InfoDim}). We are now ready to make the connection between relative information loss and information dimension in

\begin{thm}\label{thm:RILDim}
 Let $\Xvec$ be an RV with positive (finite) information dimension $\infodim{\Xvec}$. Then, if $\infodim{\Xvec|\Yvec=\yvec}$ exists and is finite $P_\Yvec$-a.s., the relative information loss equals
\begin{equation}
 \relLoss{\Xvec\to\Yvec} = \frac{\infodim{\Xvec|\Yvec}}{\infodim{\Xvec}}
\end{equation}
where $\infodim{\Xvec|\Yvec}=\int_{\dom{Y}}\infodim{\Xvec|\Yvec=\yvec}dP_\Yvec(\yvec)$.
\end{thm}

\begin{IEEEproof}
We start with Definition~\ref{def:relloss} and obtain
\begin{IEEEeqnarray}{RCL}
\relLoss{\Xvec\to\Yvec} &=& \limn \frac{\ent{\hat{\Xvec}_n|\Yvec}}{\ent{\hat{\Xvec}_n}}\\
&=& \limn \frac{\int_\dom{Y} \ent{\hat{\Xvec}_n|\Yvec=\yvec}dP_\Yvec  }{\ent{\hat{\Xvec}_n}}\\
&=& \limn \frac{\int_\dom{Y} \frac{\ent{\hat{\Xvec}_n|\Yvec=\yvec}}{\log n}dP_\Yvec  }{\frac{\ent{\hat{\Xvec}_n}}{\log n}}
\end{IEEEeqnarray}
where we divided both the numerator and the denominator by $\log n$. Assuming now the limits of the numerator and denominator exist and are finite we can continue with
\begin{IEEEeqnarray}{RCL}
 \relLoss{\Xvec\to\Yvec} &=& \frac{\limn\int_\dom{Y} \frac{\ent{\hat{\Xvec}_n|\Yvec=\yvec}}{\log n}dP_\Yvec  }{\limn\frac{\ent{\hat{\Xvec}_n}}{\log n}}\\
 &=& \frac{\limn\int_\dom{Y} \frac{\ent{\hat{\Xvec}_n|\Yvec=\yvec}}{\log n}dP_\Yvec  }{\infodim{\Xvec}}
\end{IEEEeqnarray}
where we employed Definition~\ref{def:infodim}. By assumption the limit
\begin{equation}
 \infodim{\Xvec|\Yvec=\yvec} = \limn \frac{\ent{\hat{\Xvec}_n|\Yvec=\yvec}}{\log n}
\end{equation}
exists $P_\Yvec$-a.s. Since the information dimension of an RV is upper bounded by the topological dimension of its support (e.g., by the number of coordinates of the vector $\Xvec$), one can apply Lebesgue's dominated convergence theorem (e.g.,~\cite{Rudin_Analysis3}) to exchange the order of the limit and the integral. This completes the proof.
%
\end{IEEEproof}

Before proceeding, we want to mention that the converse of Proposition~\ref{prop:infLoss} is not true\footnote{We thank an anonymous reviewer for pointing us to this fact.}, i.e., that if $\ent{\Xvec}=\infty$ and $\relLoss{\Xvec\to\Yvec}=0$, it does not necessarily follow that $\loss{\Xvec\to\Yvec}<\infty$. To this end, assume that the function $\gvec$ is such that, for all $\yvec$, $\Xvec|\Yvec=\yvec$ is a discrete RV with infinite entropy but with $\ent{\hat{\Xvec}_1|\Yvec=\yvec}<\infty$. As a consequence, the information dimension $\infodim{\Xvec|\Yvec=\yvec}=0$, which establishes \mbox{$\relLoss{\Xvec\to\Yvec}=0$}. However, 
\begin{equation}
 \loss{\Xvec\to\Yvec}=\int_\dom{Y} \ent{\Xvec|\Yvec=\yvec}dP_\Yvec(\yvec)=\infty.
\end{equation}
We give an example for such a case in the Appendix.

\section{Relative Information Loss for Functions which reduce Dimensionality}\label{sec:dimredFunct}
We now proceed to analyzing the relative information loss for measurable functions $\gvec{:}\ \dom{X}\to\dom{Y}$ which map subsets of $\dom{X}\subseteq\mathbb{R}^N$ to sufficiently well-behaved submanifolds of $\mathbb{R}^N$. In particular, let $\{\dom{X}_i\}$, $i=1,\dots,L$, denote a finite partition of $\dom{X}$ and let $P_\Xvec\ll\mu^N$, where $\mu^N$ is the $N$-dimensional Lebesgue measure. Here and throughout the remainder of this paper we assume that all involved information dimensions exist and are finite. The latter restriction is fulfilled under the mild condition that $\ent{\hat{\Xvec}_n}<\infty$~\cite{Renyi_InfoDim}, which for scalar (one-dimensional) RVs $X$ translates to $\expec{|X|^\epsilon}<\infty$ for some $\epsilon>0$~\cite{Wu_Renyi}.

We maintain

\begin{thm}\label{thm:functDimRed}
  Let $\Xvec$ be such that $P_\Xvec\ll\mu^N$ is supported on $\dom{X}\subseteq\mathbb{R}^N$ and let $\{\dom{X}_i\}$ be a partition of $\dom{X}$ such that each of its $L$ elements is a smooth $N$-dimensional manifold. Let $\gvec$ be such that $\gvec_i=\gvec|_{\dom{X}_i}$ are submersions to disjoint smooth $n_i$-dimensional manifolds $\dom{Y}_i$. Then, the relative information loss is
\begin{equation}
 \relLoss{\Xvec\to\Yvec} = \sum_{i=1}^L P_\Xvec(\dom{X}_i)\frac{N-n_i}{N}.
\end{equation}
\end{thm}

\begin{IEEEproof}
By the submersion theorem (e.g.,~\cite{Lee_SmoothManifolds}) and the fact that the images of $\dom{X}_i$ are disjoint, the preimage of every point in $\dom{Y}_i$ is a closed submanifold of $\dom{X}_i$ with codimension $n_i$. We can thus use Theorem~\ref{thm:RILDim} together with $\infodim{\Xvec}=N$ to get
\begin{multline}
 \relLoss{\Xvec\to\Yvec} = \frac{1}{N} \sum_{i=1}^L \int_{\dom{Y}_i} \infodim{\Xvec|\Yvec=\yvec}dP_\Yvec(\yvec)\\
= \frac{1}{N} \sum_{i=1}^L \int_{\dom{Y}_i} (N-n_i) dP_\Yvec(\yvec)= \sum_{i=1}^L \frac{N-n_i}{N}P_\Yvec(\dom{Y}_i)
\end{multline}
which completes the proof by noticing that $P_\Yvec(\dom{Y}_i)=P_\Xvec(\dom{X}_i)$.
%
%
\end{IEEEproof}

Before proceeding, two interesting facts are worth mentioning: First of all, as a submersion is a smooth mapping between smooth manifolds, we can state
\begin{cor}\label{cor:dim}
If $P_\Xvec\ll\mu^N$ and if $\gvec$ is as in Theorem~\ref{thm:functDimRed}, then
 \begin{equation}
 \relLoss{\Xvec\to\Yvec} = 1-\frac{\infodim{\Yvec}}{\infodim{\Xvec}}
\end{equation}
provided the information dimension of $\Yvec$ exists.
\end{cor}

\begin{IEEEproof}
 We note from the proof of Theorem~\ref{thm:functDimRed} that
\begin{equation}
 \relLoss{\Xvec\to\Yvec} = 1- \sum_{i=1}^L \frac{n_i}{N} P_\Yvec(\dom{Y}_i)
\end{equation}
where $N=\infodim{\Xvec}$. We now show that $\infodim{\Yvec|\Xvec\in\dom{X}_i}=n_i$. By the fact that $\gvec|_{\dom{X}_i}\in\mathcal{C}^\infty$ is a submersion, the preimages of $\mu^{n_i}$-null sets are $\mu^N$-null sets themselves~\cite{Ponomarev_Preimage}. Consequently, since $P_\Xvec\ll\mu^N$, the conditional probability measure supported on $\dom{Y}_i$ is absolutely continuous w.r.t. the $n_i$-dimensional Lebesgue measure ($P_{\Yvec|\Xvec\in\dom{X}_i}\ll\mu^{n_i}$). Since $\dom{Y}_i$ is a smooth $n_i$-dimensional manifold, we again obtain with R\'{e}nyi~\cite{Renyi_InfoDim} that $\infodim{\Yvec|\Xvec\in\dom{X}_i}=n_i$. The proof is complete with~\cite[Theorem~2]{Wu_Renyi} or~\cite{Smieja_EntropyMixture}, noting that
\begin{equation}
 \infodim{\Yvec} = \sum_{i=1}^L \infodim{\Yvec|\Xvec\in\dom{X}_i} P_\Xvec(\dom{X}_i)=\sum_{i=1}^L n_i P_\Yvec(\dom{Y}_i).\label{eq:avgDim}
\end{equation}
\end{IEEEproof}

\ifthenelse{\arxiv=1}{In fact, already R\'{e}nyi showed~\eqref{eq:avgDim} for a mixture of a one-dimensional continuous RV and a discrete RV in~\cite{Renyi_InfoDim}.}{} For more complicated measures (e.g., measures with a non-integer information dimension) or more general functions it might be hard to prove Theorem~\ref{thm:functDimRed} and Corollary~\ref{cor:dim}. We conjecture, however, that the conditions imposed in Theorem~\ref{thm:functDimRed} can be loosened such that at least for a larger class of functions (e.g., those for which the images of $\dom{X}_i$ are not disjoint) the relative information loss can be evaluated.

Secondly, we believe that one should be able to analyze cascaded systems, at least in a restricted class of cases: To this end, assume that we have a cascade of two projections: one on the first two and then one on the first coordinate, i.e., $\gvec{:}\ [X_1,X_2,X_3]\to[X_1,X_2]$ and $h{:}\ [X_1,X_2]\to X_1$. Let the RVs $X_1$, $X_2$, and $X_3$ have a continuous joint distribution. Then, according to Theorem~\ref{thm:functDimRed}, the function $\gvec$ destroys one third of the information, while $h$ destroys half of the remaining information. In total, two thirds of the information are lost. Indeed, we obtain
\begin{equation}
 \relLoss{[X_1,X_2,X_3]\to X_1} = \frac{1}{3}+\frac{1}{2}-\frac{1}{3}\frac{1}{2}=\frac{2}{3}.
\end{equation}
We are thus lead to the following
\begin{con}\label{con:cascade}
  Let $\gvec$ be a (measurable) function describing a system with input $\Xvec$ and output $\Yvec$. Let further $\hvec$ be another such system function which transforms input $\Yvec$ to output $\Zvec$. There exist conditions under which
\begin{equation}
 \relLoss{\Xvec\to\Zvec} = \relLoss{\Xvec\to\Yvec}+\relLoss{\Yvec\to\Zvec} - \relLoss{\Xvec\to\Yvec}\relLoss{\Yvec\to\Zvec}.\label{eq:cascade}
\end{equation}
\end{con}
Clearly, our example of the cascade of two projections fulfills these conditions. Moreover, as we show in\ifthenelse{\arxiv=1}{ the Appendix}{~\cite{Geiger_RILPCA_arXiv}}, this conjecture holds for discrete RVs with finite entropy. 


Finally, it is worth mentioning that the shape of the distribution has no effect on the relative amount of information lost. It is essentially this behavior which leads to the somewhat counter-intuitive results presented in the following sections.

\section{Toy Examples for Dimensionality Reduction}\label{ssec:Toy}
With the help of a few simple examples we now try to make the operational meaning of relative information loss intuitive. At the same time we highlight its importance to the development of an information-centered system theory. 

Let us introduce a two-dimensional RV $\Xvec$, $P_\Xvec\ll\mu^2$ and $\infodim{\Xvec}=2$, and a transform simply adding the two vector components, i.e.,
\begin{equation}
 Y=g(\Xvec)=X_1+X_2.
\end{equation}
We can represent this transform by a cascade of an invertible linear transform (e.g., $\underline{\mathbf{T}}{:}\ \Xvec\to[X_1+X_2,X_1]$) and a simple projection onto the first component. Let now $\tilde{\Xvec}=[X_1+X_2,X_1]$. Since $\underline{\mathbf{T}}$ is invertible and linear, it is bi-Lipschitz and thus preserves the dimension of the transformed RVs. Thus,
\begin{equation}
 \relLoss{\Xvec\to Y} = \relLoss{\tilde{\Xvec}\to Y}=0.5
\end{equation} 
by Theorem~\ref{thm:functDimRed}. Note that this is another example where Conjecture~\ref{con:cascade} holds.

With this result and the underlying theorems, the theory of information-processing systems can be extended from treating cascade structures as in~\eqref{eq:cascade} to also considering parallel structures whose outputs are added. This were not possible by considering only absolute information loss (e.g.,~\cite{Geiger_ILStatic_IZS}), as it would be infinite in these cases.

\begin{figure}[t!]
  \begin{center}
 \begin{pspicture}[showgrid=false](-2,-2)(2,2)
	\psaxeslabels{->}(0,0)(-2,-2)(2,2){$x$}{$g(x)$}
  \psplot[style=Graph,linecolor=black,plotpoints=500]{-1.8}{-0.85}{x}
	\psplot[style=Graph,linecolor=black,plotpoints=500]{0.85}{1.8}{x}
	\psplot[style=Graph,linecolor=black,plotpoints=500]{-0.8}{.8}{0}
	\psdisk[fillcolor=black](0.8,0){0.07}\psdisk[fillcolor=black](-0.8,0){0.07}
	\pscircle(0.8,0.8){0.07}\pscircle(-0.8,-0.8){0.07}
	\psTick{90}(0.8,0) \rput(0.8,-0.3){$c$}
	\psTick{90}(-0.8,0) \rput(-0.8,-0.3){$-c$}
 \end{pspicture}
\end{center}
\caption{The center clipper -- another example for dimensionality reduction.}
\label{fig:cclipper}
\end{figure}
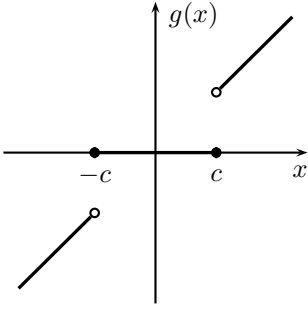

As a second toy example we consider a center clipper (see Fig.~\ref{fig:cclipper}), which is commonly used for noise suppression or residual echo cancellation\ifthenelse{\arxiv=1}{~\cite{Vary_DigSpeechTransm}}{}. We describe the center clipper by the following function:
\begin{equation}
 g(x) = \begin{cases}
 x,& \text{ if }|x|>c\\ 0, & \text{ else}
        \end{cases}.
\end{equation}
Clearly, the domain of this function can be partitioned into three elements, upon two of them the function is the identity function. On the third set, $[-c,c]$, the function is a submersion to a zero-dimensional manifold. We can thus apply Theorem~\ref{thm:functDimRed} and obtain $\relLoss{X\to Y}=P_X([-c,c])$.

It is interesting to see that in both examples the distribution of the input signal does not have an influence on the relative information loss, as long as it is continuous. This is counter-intuitive in the sense that adding two strongly correlated RVs should preserve more information than adding two independent RVs, or, in the sense that center clipping a large signal should not hurt too much. This intuition, however, is based on the fact that one tends to attribute unequal importance to different aspects of the information contained in the input signal (e.g., principle direction, magnitude, etc.).

\section{PCA with Population Covariance Matrix}\label{ssec:PCA}
In PCA one uses the eigenvalue decomposition (EVD) of the covariance matrix of a multivariate input to obtain a different representation of the input vector. In particular, let $\Xvec$ be an RV with distribution $P_\Xvec\ll\mu^N$ and information dimension $\infodim{\Xvec}=N$. We further assume that $\Xvec$ has zero mean and a positive definite population covariance matrix $\Cx=\expec{\Xvec\Xvec^T}$ which is known a priori. The case where $\Cx$ is not known but has to be estimated from the data is considered in Section~\ref{sec:pcaloss}.

The EVD of the covariance matrix yields
\begin{equation}
 \Cx = \underline{\mathbf{W}}\underline{\boldsymbol{\Sigma}}\underline{\mathbf{W}}^T
\end{equation}
where $\underline{\mathbf{W}}$ is an orthogonal matrix (i.e., $\underline{\mathbf{W}}^{-1}=\underline{\mathbf{W}}^T$) and $\underline{\boldsymbol{\Sigma}}$ is a diagonal matrix consisting of the $N$ eigenvalues of $\Cx$.
%
We now can describe the PCA by the following linear transform:
\begin{equation}
 \Yvec=\gvec(\Xvec)=\underline{\mathbf{W}}^T\Xvec.\label{eq:PCA}
\end{equation}
As in Section~\ref{ssec:Toy}, the linear transform is bi-Lipschitz, and the information loss vanishes\footnote{Not only the relative information loss {$\relLoss{\Xvec\to\Yvec}=0$}, but also the absolute information loss {$\loss{\Xvec\to\Yvec}=0$}.}.

Often, however, the PCA is used for dimensionality reduction, where after the linear transform in~\eqref{eq:PCA} the elements of the random vector $\Yvec$ with the smallest variances are discarded (thus, preserving the subspace with the largest variance).  Essentially, the mapping from $\Yvec$ to, e.g.,
\begin{equation}
 \Yvec_M = [Y_1,\dots,Y_M]
\end{equation}
is a projection onto the first $M<N$ coordinates, which is a submersion between two smooth manifolds. We can thus apply Theorem~\ref{thm:functDimRed} and obtain the relative information loss $\relLoss{\Yvec\to\Yvec_M}=\frac{N-M}{N}$. In analogy with the example in Section~\ref{ssec:Toy} we thus get
\begin{equation}
 \relLoss{\Xvec\to\Yvec_M}=\frac{N-M}{N}.
\end{equation}

We can now extend this analysis to the case where from $\Yvec_M$ an ($N$-dimensional) estimate $\tilde{\Xvec}_M$ of the original data $\Xvec$ is reconstructed. This estimate is obtained using a linear transform
\begin{equation}
 \tilde{\Xvec}_M = \underline{\mathbf{W}}\underline{\mathbf{I}}_M^T \Yvec_M
\end{equation}
where $\underline{\mathbf{I}}_M$ is a rectangular identity matrix with $M$ rows and $N$ columns. The (full-rank) matrix $\underline{\mathbf{I}}_M^T$ is a mapping to a higher-dimensional space (i.e., from $\mathbb{R}^M$ to $\mathbb{R}^N$) and is thus bi-Lipschitz; so is the rotation with the matrix $\underline{\mathbf{W}}$. Furthermore, the transform from $\Yvec_M$ to $\tilde{\Xvec}_M$ is invertible and, as a consequence, no additional information is lost. We thus state
\begin{equation}
 \relLoss{\Xvec\to\tilde{\Xvec}_M} = \frac{N-M}{N}
\end{equation}
where, using above notation, $\tilde{\Xvec}_M = \underline{\mathbf{W}}\underline{\mathbf{I}}_M^T \underline{\mathbf{I}}_M\underline{\mathbf{W}}^T\Xvec$.

Indeed, the same result would have been obtained if the rotation would have been performed using any other orthogonal matrix and regardless which elements of the rotated vector were discarded. In particular, also if just the first $M$ components of $\Xvec$ would have been preserved, we would have $\relLoss{\Xvec\to\Xvec_M}=\frac{N-M}{N}$.

Of course, PCA is known to be optimal in the sense that, by discarding the elements of $\Yvec$ with the smallest variances, the mean-squared reconstruction error for $\tilde{\Xvec}_M$ is minimized~\cite{Deco_ITNN}. For this interpretation, measuring the information loss \emph{with respect to a relevant random variable} may do the trick, providing us with a statement about the optimality of PCA in an information-theoretic sense (cf.~\cite{Plumbley_TN,Deco_ITNN}). Conversely, if one cannot determine which information at the input is relevant, one has no reason to perform PCA prior to reducing the dimension of the data.

\section{PCA with Sample Covariance Matrix}\label{sec:pcaloss}
We argued in Section~\ref{ssec:PCA} that the PCA without dimensionality reduction is an invertible transform. This, however, is only true if one has access to the orthogonal matrix $\underline{\mathbf{W}}$; if not, i.e., if one feeds a system with the data matrix $\Xmat$ and just receives the rotated matrix $\Ymat$ (see Fig.~\ref{fig:PCA}), it can be shown that information is lost. In this section we make this statement precise.

We believe that the following analysis can be generalized to all data matrices with a continuous joint distribution, i.e., $P_{\Xmat}\ll\mu^{nN}$. For the sake of simplicity, in this paper, we will focus on a particularly simple scenario: Let $\Xmat$ denote a matrix where each of its $n$ columns represents an \emph{independent} sample of an $N$-dimensional Gaussian RV $\Xvec$. Again, let $\Xvec$ have zero mean and positive definite population covariance matrix $\Cx$. As a consequence, the probability distribution of the data matrix $\Xmat$ is absolutely continuous w.r.t. the $nN$-dimensional Lebesgue measure ($P_{\Xmat}\ll\mu^{nN}$).

\newcommand{\Wh}{\underline{\hat{\mathbf{W}}}}
\newcommand{\Sh}{\underline{\hat{\boldsymbol{\Sigma}}}}
\begin{figure}[t]
 \centering  
\begin{pspicture}[showgrid=false](1.25,-0.75)(8.75,3.5)
	\psset{style=RoundCorners,style=Arrow}
 	\pssignal(1.25,2){x}{$\Xmat$}
	\pssignal(7.75,2){sigma}{$\Sh$}
	\pssignal(8.75,0){y}{$\Ymat$}
	\dotnode(2.5,2){dotx}
	\psfblock[framesize=1.2 .8](3.5,2){cov}{Cov}
	\psfblock[framesize=1.2 .8](6.25,2){evd}{EVD}
	\psfblock[framesize=1.5 1](5,0){w}{$\Wh$}

	\nclist{ncline}[naput]{x,dotx,cov,evd $\frac{1}{n}\Xmat\Xmat^T$,sigma}

	\ncline{w}{y}
	\ncangle[angleA=-90,angleB=-180]{dotx}{w}
	\ncangle[angleA=-90,angleB=90]{evd}{w}\nbput{$\Wh$}

	\fnode[style=Dash,linecolor=gray,framesize=6 3.5](5,1){box}
	\nput{90}{box}{PCA}
\end{pspicture}
\caption{The PCA as a nonlinear input-output system. ``Cov'' denotes the computation of the sample covariance matrix and ``EVD'' stands for eigenvalue decomposition.}
\label{fig:PCA}
\vspace*{-0.5cm}
\end{figure}
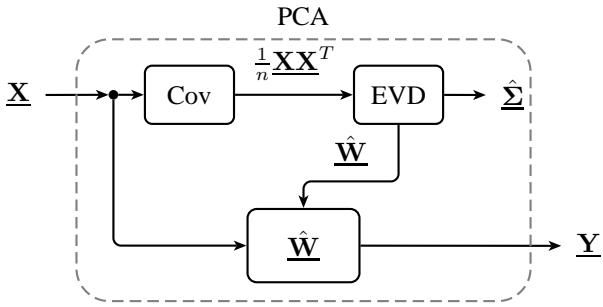

The sample covariance matrix $\Chx=\frac{1}{n}\Xmat\Xmat^T$ is symmetric and almost surely positive definite. In the usual case where $n\geq N$ one can show that $\frac{N(N+1)}{2}$ entries can be chosen and that the remaining entries depend on these in a deterministic manner. Indeed, since in this case the distribution of $\Chx$ possesses a density (the Wishart distribution, cf.~\cite{Muirhead_Multivariate}), the distribution is absolutely continuous w.r.t. the Lebesgue measure on an $\frac{N(N+1)}{2}$-dimensional submanifold of the $N^2$-dimensional Euclidean space. With some abuse of notation we thus write $P_{\Chx}\ll\mu^{\frac{N(N+1)}{2}}$.

The orthogonal matrix $\Wh$ for PCA (see~\eqref{eq:PCA}; now applied to the matrix $\Xmat$ instead of the vector $\Xvec$) is obtained from the EVD of the sample covariance matrix, i.e.,
\begin{equation}
 \Chx=\Wh\Sh\Wh^T
\end{equation}
where $\Sh$ is the diagonal matrix containing the eigenvalues of $\Chx$. The joint distribution of the $N$ eigenvalues of $\Chx$ possesses a density~\cite[Ch.~9.4]{Muirhead_Multivariate}; thus, the distribution of $\Sh$ is absolutely continuous w.r.t. the Lebesgue measure on an $N$-dimensional submanifold of the $N^2$-dimensional Euclidean space, or $P_{\Sh}\ll\mu^N$. 

Clearly, the entries of $\Chx$ are smooth functions of the eigenvalues and the entries of $\Wh$. Images of Lebesgue null-sets under smooth functions between Euclidean spaces of same dimension are null-sets themselves; were the probability measure $P_{\Wh}$ supported on some set of dimensionality lower than $\frac{N(N-1)}{2}$, the image of the product of this set and $\mathbb{R}^N$ (for the eigenvalues) would be a Lebesgue null-set with positive probability measure. Since this contradicts the fact that $\Chx$ is continuously distributed, it follows that 
\begin{equation}
 P_{\Wh}\ll\mu^{\frac{N(N-1)}{2}}.
\end{equation}
We now argue\ifthenelse{\arxiv=1}{ (see Appendix for a rigorous discussion)}{ (in~\cite{Geiger_RILPCA_arXiv} more rigorously than here)} that the rotated data does not tell us anything about the rotation, hence
\begin{equation}
 P_{\Wh|\Ymat=\underline{\yvec}}\ll\mu^{\frac{N(N-1)}{2}}.
\end{equation}

Knowing $\Ymat=\ymat$, $\Xmat$ is a linear, bi-Lipschitz function of $\Wh$, and thus the information dimension remains unchanged. Therefore we get with Theorem~\ref{thm:RILDim}
\begin{equation}
 \relLoss{\Xmat\to\Ymat} = \frac{\infodim{\Xmat|\Ymat}}{\infodim{\Xmat}} = \frac{N(N-1)}{2nN}=\frac{N-1}{2n}.
\end{equation}

%

We now drop a little of the mathematical rigor to analyze, for the sake of completeness, the less common case where there are less data samples than there are dimensions for each sample ($n<N$). In this case, the sample covariance matrix is not full rank, which means that the EVD yields $N-n$ vanishing eigenvalues. Assuming that still $P_{\Chx}\ll\mu^{n\frac{2N-n+1}{2}}$ one finds along the same lines as in the case $n\geq N$ that the loss evaluates to
\begin{equation}
 \relLoss{\Xmat\to\Ymat} = \frac{2N-n-1}{2N}.
\end{equation}
The behavior of the relative information loss as a function of $n$ is shown in Fig.~\ref{fig:PCAloss} for different choices of $N$.\ifthenelse{\arxiv=1}{ A worked example for a particularly simple case can be found in the Appendix.}{}

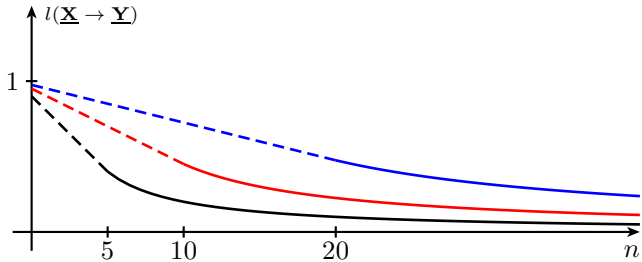
\begin{figure}[t]
 \centering  
\begin{pspicture}[showgrid=false](-0.5,-0.5)(8,3)
	\psset{style=RoundCorners}
	\psaxeslabels[style=Arrow](0,0)(-0.25,-0.25)(8,3){$n$}{\scriptsize $\relLoss{\Xmat\to\Ymat}$}
	\psplot[style=Graph,linewidth=1pt]{1}{8}{x 5 mul 2 div -1 exp 2 mul}
	\psplot[style=Graph,linewidth=1pt,linecolor=red]{2}{8}{x 5 mul 9 div 2 mul -1 exp 2 mul}
	\psplot[style=Graph,linewidth=1pt,linecolor=blue]{4}{8}{x 5 mul 19 div 2 mul -1 exp 2 mul}
	\psset{style=Dash}
	\psplot[style=Graph,linewidth=1pt]{0}{1}{x -5 mul 9 add 10 div 2 mul}
	\psplot[style=Graph,linewidth=1pt,linecolor=red]{0}{2}{x -5 mul 19 add 20 div 2 mul}
	\psplot[style=Graph,linewidth=1pt,linecolor=blue]{0}{4}{x -5 mul 39 add 40 div 2 mul}
	\psTick{0}(0,2) \psTick{90}(1,0)\psTick{90}(2,0)\psTick{90}(4,0)
	\rput(-0.25,2){$1$} \rput(1,-0.25){$5$}\rput(2,-0.25){$10$}\rput(4,-0.25){$20$}
\end{pspicture}
\caption{Relative information loss in the PCA with sample covariance matrix as a function of the number $n$ of independent measurements. The cases $N=5$ (black), $N=10$ (red), and $N=20$ (blue) are shown. The dashed lines indicate the conjectured loss for singular sample covariance matrices.}
\label{fig:PCAloss}\vspace*{-0.5cm}
\end{figure}

The relative information loss induced by PCA results from the fact that one cannot know \emph{which rotation} led to the output data matrix. As a consequence, the relative information loss decreases with a larger number of samples: the total information increases while the uncertainty about the rotation remains the same. Note further that the relative information loss cannot exceed $\frac{N-1}{N}$ (for $n=1$), which is due to the fact that the rotation preserves the norm of the sample. 

\ifthenelse{\arxiv=1}{The PCA with sample covariance matrix also allows different interpretations of information loss, in addition to the information lost in the rotation: First, as we show in the Appendix, by the fact that the sample covariance matrix of $\Ymat$ is a diagonal matrix, the possible values of $\Ymat$ are restricted to a submanifold of dimensionality smaller than $nN$. Naturally, this also restricts the amount of information which can be conveyed in the output. In contrary to this, in PCA using the population covariance matrix the sample covariance matrix of $\Ymat$ will (almost surely) do not contain zeros, so the entries of $\Ymat$ will not be restricted deterministically.}{}

Finally, it is interesting to observe that in this case the \emph{absolute} information loss in the PCA is infinite\ifthenelse{\arxiv=1}{ (see~Proposition~\ref{prop:infLoss})}{~\cite{Geiger_RILPCA_arXiv}}, even if no additional dimensionality reduction is performed. Moreover, this analysis not only holds for the PCA, but for \emph{any} rotation which depends on the input data in a similar manner --  in this sense, the PCA is not better than any other rotation.

\section{Discussion}\label{sec:relevant}
Relative information loss was introduced to cope with the shortcomings of absolute information loss, especially in cases where the information loss $\loss{\Xvec\to\Yvec}$ and the information transfer $\mutinf{\Xvec;\Yvec}$ are infinite. However, as the examples of Section~\ref{ssec:Toy} and~\ref{ssec:PCA} showed, even relative information loss not necessarily provides full insight and intuitive interpretations in some cases.

These cases may be characterized by the fact that not all information at the input of the system is \emph{relevant} for a given application, i.e., we are actually not interested in $\Xvec$ itself but in some random variable $\Zvec$ somehow related to it. $\Zvec$, however, is not accessible directly, but only through a function $\gvec$ of the related RV $\Xvec$. The logical consequence is thus to define a quantity which captures the information loss relevant w.r.t. $\Zvec$. We did so in~\cite{Geiger_Relevant_arXiv}, where we analyzed some of this quantity's properties and made a connection to the signal enhancement problem.
%
%
%
%

A very similar quantity has already been introduced by Plumbley~\cite{Plumbley_TN} in the context of unsupervised learning. Using this quantity he proved the optimality of the PCA under appropriate constraints, and argued that information loss in some cases is more versatile than mutual information.

In order to be able to minimize this \emph{relevant information loss}, one clearly has to know something about the relevant RV $\Zvec$ and its relationship to the system input $\Xvec$. If one does not have this knowledge, all bits of information have to be treated equally, leading to our notion of (relative) information loss. As a direct consequence, while in some cases the PCA prior to dimensionality reduction might be the optimal solution (cf.~\cite{Plumbley_TN,Geiger_Relevant_arXiv}), without knowledge about the relevant information one has no reason to apply it to a data set. Even more so, as we showed in Section~\ref{sec:pcaloss}, even without reducing the dimensionality, information can be lost, which should prevent one from unjustified use of PCA.

\ifthenelse{\arxiv=1}{
\begin{figure*}[t]
 \centering  
\begin{pspicture}[showgrid=false](1.25,-1)(14,5)
  \psset{style=RoundCorners,style=Arrow}
  \pssignal(1.25,3){x}{$\Xmat$}
  \pssignal(14,3){wout}{$\Wh$}
  \pssignal(14,4){sigma}{$\Sh$}
  \pssignal(14,0){y}{$\Ymat$}

  \dotnode(2.5,3){dotx}
  \psfblock[framesize=1.2 .8](3.5,3){cov}{Cov}
  \psfblock[framesize=1.2 .8](6.25,3){evd}{EVD}
  \psfblock[framesize=1.5 1](6.25,0){w}{$\Wh$}
  \dotnode(8,3){dotw}
  \dotnode(11,4){dots}
  \dotnode(11,0){doty}
 
  \ncline{x}{cov}\naput{1}
  \ncline{cov}{evd}\nbput{$\frac{N+1}{2n}$}\naput{$\Chx$}
  \ncline{evd}{wout}\nbput{$\frac{N-1}{2n}$}
  \ncline{w}{y}\nbput{$\frac{2n-N+1}{2n}=1-\frac{N-1}{2n}$}
  \ncangle[angleA=-90,angleB=-180]{dotx}{w}\naput{1}
  \ncangle[angleA=90,angleB=180]{evd}{sigma}\naput{$\frac{1}{n}$}
  \ncangle[angleA=-90,angleB=90]{dotw}{w}\nbput{$\frac{N-1}{2n}$}
  \psset{linecolor=gray}
  \pssignal(14,1.5){yt}{\textcolor{gray}{$\tilde{\Ymat}$}}
  \psfblock[framesize=1.5 1](11,1.5){sigmult}{\textcolor{gray}{$\Sh^{-1/2}$}}
  \ncline{sigmult}{yt}\naput{\textcolor{gray}{$\frac{2n-N-1}{2n}$}}
  \ncline{doty}{sigmult}
  \ncline{dots}{sigmult}
\end{pspicture}
\caption{Information propagation in PCA with sample covariance matrix. The values on the arrows indicate the relative information transfer $t(\Xmat\to\cdot)$. We assume $n>N$ in this case. The gray part considers the sphering transform (see text). Note that the separate information transfers to $\Sh$, $\Wh$, and $\Ymat$ add up to $1+\frac{1}{n}>1$. This is because the information in $\Sh$ is already contained in $\Ymat$. Note further that $t(\Xmat\to\Ymat,\Sh,\Wh)=1$, as expected.}
\label{fig:PCAbudget}
\end{figure*}
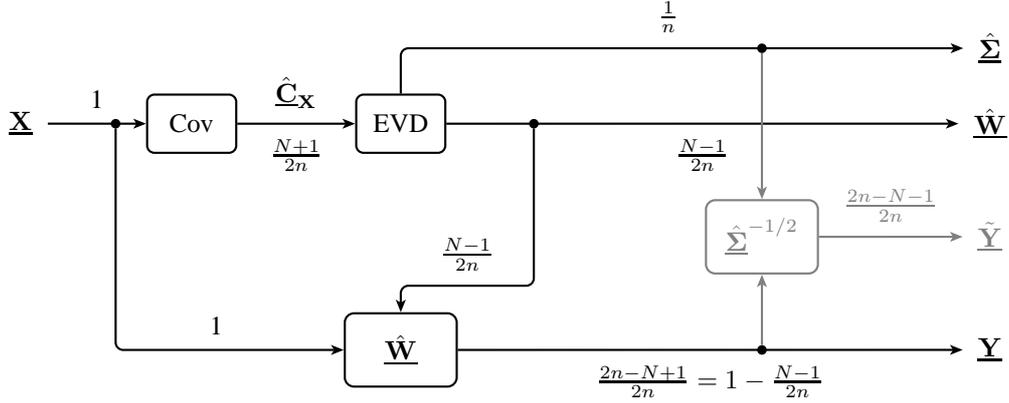
We now discuss a different aspect of our work about information loss -- be it absolute or relative -- in deterministic systems. The definition of these quantities allows us to quantify and, hopefully, understand the propagation of information in a network of systems. 

Take, for example, the PCA: The information at the input is split and propagates through the system, with parts of it being lost in some paths and preserved in others. By denoting $t(\Xmat\to\cdot)=1-\relLoss{\Xmat\to\cdot}$ the relative information transfer, we can thus redraw the system model and obtain Fig.~\ref{fig:PCAbudget} where the arrows are labeled according to the relative information transfer along them. In particular, the eigenvalue decomposition splits the information contained in the covariance matrix into a part describing the eigenvectors and a part describing the eigenvalues. The former part is lost in PCA, while the second is preserved in the output. However, as we already mentioned before, the rotation (i.e., the multiplication with the orthogonal matrix) removes exactly as much information from the input as is contained in the orthogonal matrix. The information contained in both $\Ymat$ and $\Wh$ suffices to reconstruct $\Xmat$. If we perform in addition a sphering transform on $\Ymat$ (thus making all eigenvalues unity), one needs the triple $\tilde{\Ymat}$, $\Wh$ and $\Sh$ to reconstruct $\Xmat$.

It is obvious that this transfer graph must not be understood in the sense of a preservation theorem, similar to Kirchhoff's current law: Information can be split and fused, and after splitting the sum of information at the output needs not equal the sum of information at the input (as it is, by coincidence, for the EVD). If the output information is less than the input information, the intermediate system destroyed the remaining information (as in rotation). Conversely, the if the sum of output information exceeds the input information, this only means that parts of these information must be the same (as, e.g., $\Ymat$ includes the knowledge about $\Sh$).
}{}

\section{Conclusion}
We have introduced the notion of relative information loss to analyze systems for which the output has a lower dimension than the input, exploiting a tight connection to R\'{e}nyi's information dimension. As a first result, we showed that the relative information loss for dimensionality reduction is not affected by performing a principle component analysis (PCA) beforehand.

We then showed that even without dimensionality reduction the PCA is not information lossless, given that the sample covariance matrix is used to compute the rotation matrix. There, the relative information loss appears to decrease with increasing sample size.

We proposed to use these somewhat counter-intuitive results as motivation to introduce a notion of information loss which takes the relevance of the data into account, and thus connects to results available in the literature. A detailed analysis of this relevant information loss is within the scope of future work.

\section*{Acknowledgment}
The authors thank Yihong Wu, Wharton School, University of Pennsylvania, for bringing R\'{e}nyi's information dimension to our attention.

\ifthenelse{\arxiv=1}{
\appendices

\section{Example for infinite absolute and vanishing relative information loss}
We now give an example where an infinite amount of information is lost (i.e., $\loss{X\to Y}=\infty$), but for which the relative information loss nevertheless vanishes (i.e., $\relLoss{X\to Y}=0$). To this end, assume that we consider the following scalar function $g{:}\ (0,1]\to(0,1]$
\begin{equation}
 g(x) = 2^n (x-2^{-n}) \text{ if } x\in(2^{-n},2^{-n+1}],\ n\in\mathbb{N}
\end{equation}
In other words, this function maps every interval $(2^{-n},2^{-n+1}]$ onto the interval $(0,1]$. Assume now further that the input variable $X$ has a continuous distribution with density function
\begin{multline}
 f_X(x) = 2^n \left(\frac{1}{\log(n+1)}-\frac{1}{\log(n+2)}\right),\\\text{ if } x\in(2^{-n},2^{-n+1}],\ n\in\mathbb{N}
\end{multline}
where $\log$ denotes the binary logarithm. As an immediate consequence, the output RV $Y$ is uniformly distributed on $(0,1]$. Furthermore, the information dimension of the input is given by $\infodim{X}=1$.

Since the function is piecewise strictly monotone we can apply the reasoning of~\cite{Geiger_ILStatic_IZS} and claim that
\begin{equation}
\loss{X\to Y}=\ent{X|Y}=\ent{W|Y}
\end{equation}
where
\begin{equation}
 W= n \text{ if } X\in(2^{-n},2^{-n+1}].
\end{equation}
In particular, for the given input density function we obtain
\begin{multline}
 \Prob{W=n|Y=y}= \Prob{W=n}\\=\frac{1}{\log(n+1)}-\frac{1}{\log(n+2)}.
\end{multline}
For this distribution, however, it is known that the entropy is infinite, i.e., $\ent{W}=\infty$~\cite{Baer_InfEntropy}. This shows that in this case the absolute information loss is infinite as well.

However, for every $y\in(0,1]$ the preimage is a countable set, thus $X|Y=y$ is a discrete RV. In addition to that, since $X$ is supported on a compact set, so is $X|Y=y$, and every quantization $\hat{X}_n|Y=y$ will have finite entropy. Thus, $\infodim{X|Y=y}=0$ for all $y$~\cite{Renyi_InfoDim}, and with Theorem~\ref{thm:RILDim} we obtain $\relLoss{X\to Y}=0$.

\section{Proof of Conjecture~\ref{con:cascade} for discrete RVs}
Let $\Xvec$, $\Yvec$, and $\Zvec$ be discrete RVs with finite entropy, and let further $\Yvec=\gvec(\Xvec)$ and $\Zvec=\hvec(\Yvec)$. As an immediate result, $\mutinf{\Xvec;\Yvec}=\ent{\Yvec}$ and $\mutinf{\Yvec;\Zvec}=\ent{\Zvec}$. We next note that
\begin{equation}
 \relLoss{\Xvec\to\Zvec} = \frac{\ent{\Xvec|\Zvec}}{\ent{\Xvec}}
\end{equation}
since the RVs are already discrete. Furthermore, we can write
\begin{equation}
 \relLoss{\Xvec\to\Zvec} = \frac{\ent{\Xvec}-\mutinf{\Xvec;\Zvec}}{\ent{\Xvec}} =1-t(\Xvec\to\Zvec)
\end{equation}
where $t(\Xvec\to\Zvec)=\frac{\mutinf{\Xvec;\Zvec}}{\ent{\Xvec}}=\frac{\ent{\Zvec}}{\ent{\Xvec}}$, since also $\Zvec$ is a function of $\Xvec$. Expanding the latter term -- the \emph{relative information transfer} -- one obtains
\begin{multline}
 1-\relLoss{\Xvec\to\Zvec} = t(\Xvec\to\Zvec) = \frac{\ent{\Zvec}}{\ent{\Xvec}}\frac{\ent{\Yvec}}{\ent{\Yvec}} \\= \frac{\mutinf{\Yvec;\Zvec}}{\ent{\Yvec}}\frac{\mutinf{\Xvec;\Yvec}}{\ent{\Xvec}} = t(\Xvec\to\Yvec)t(\Yvec\to\Zvec)\\
=(1-\relLoss{\Xvec\to\Yvec})(1-\relLoss{\Yvec\to\Zvec})\\
= 1-\relLoss{\Xvec\to\Yvec}-\relLoss{\Yvec\to\Zvec}+\relLoss{\Xvec\to\Yvec}\relLoss{\Yvec\to\Zvec}.
\end{multline}
Rearranging completes the proof.\endproof
We note in passing that this result holds for the cascade of \emph{deterministic} systems and that a generalization to Markov chains $\Xvec-\Yvec-\Zvec$ is not possible.

\section{Information dimension of $\Xmat|\Ymat=\underline{\yvec}$}
We already observed that $P_{\Wh}\ll\mu^{\frac{N(N-1)}{2}}$. Furthermore, since we know that the sample covariance matrix of $\Ymat$ is a diagonal matrix, the corresponding equations
\begin{equation}
 \forall1\leq i< j\leq N:\quad (\Cy)_{ij} = \frac{1}{n}\sum_{k=1}^n Y_{ik}Y_{jk}=0
\end{equation}
restrict the possible values of $\Ymat$ from an $nN$-dimensional to an $M$-dimensional subspace with
\begin{equation}
 M=nN-\frac{N(N-1)}{2}.
\end{equation}
In fact, it can be shown that $M$ elements of $\Ymat$ are random while the remaining $(nN-M)$ depend on these in a deterministic manner\footnote{E.g., one could determine $Y_{ij}$ from the equation for $(\Cy)_{ij}$.}. 

Since $\Xmat=\Wh\Ymat$, $\Xmat$ is a smooth function from $\mathbb{R}^{\frac{N(N-1)}{2}}\times\mathbb{R}^M$ (the ranges of the values of $\Wh$ and $\Ymat$) to $\mathbb{R}^{nN}$ (the range of values of $\Xmat$), thus from $\mathbb{R}^{nN}$ to $\mathbb{R}^{nN}$. Since this smooth mapping maps null-sets to null-sets~\cite[Lem.~6.2]{Lee_SmoothManifolds}, we obtain
\begin{equation}
 P_{(\Wh,\Ymat)}\ll\mu^{nN}.
\end{equation}
(We are well aware that $\Wh$ and $\Ymat$ together have more than $nN$ entries, but only $nN$ of those can be chosen freely. In other words, the graph of the functions defining the remaining entries of $\Ymat$ and $\Wh$ is an $nN$-dimensional submanifold of $\mathbb{R}^{N(n+N)}$~\cite[Lem.~5.9]{Lee_SmoothManifolds}.)
The joint distribution of $(\Wh,\Ymat)$ thus possesses a density, and by marginalizing and conditioning so does $\Wh|\Ymat=\ymat$. As a consequence,
\begin{equation}
 P_{\Wh|\Ymat=\ymat}\ll\mu^{\frac{N(N-1)}{2}}.
\end{equation}
Note further that this does not mean that $\Wh$ is independent of $\Ymat$ -- it just means that these two quantities are at least not related deterministically.

The final step is taken by recognizing that if one knows $\Ymat=\ymat$, then $\Xmat|\Ymat=\ymat$ is a linear function of $\Wh|\Ymat=\ymat$. Since $\Ymat$ has full rank, the linear function maps the $N^2$-dimensional space of $\Wh$ (on which the probability mass is concentrated on an $\frac{N(N-1)}{2}$-dimensional subspace) to the $N^2$-dimensional linear subspace of $\mathbb{R}^{nN}$. With~\cite[Remark~28.9]{Yeh_RealAnalysis} this transform is bi-Lipschitz and preserves the information dimension. Thus,
\begin{equation}
 \infodim{\Xmat|\Ymat=\ymat}=\frac{N(N-1)}{2}.
\end{equation}

\section{Example: PCA with singular Sample Covariance Matrix}\label{sec:examplePCA}
We now give a worked example for the -- admittedly less common -- case of a singular sample covariance matrix. Let $\Xvec$ be a two-dimensional Gaussian RV, and let $n=1$, i.e., $\Xmat=\Xvec$. The sample covariance matrix is given by
\begin{equation}
 \Chx = \left[\begin{array}{cc}
               X_1^2&X_1X_2\\X_1X_2&X_2^2
              \end{array}
\right]
\end{equation}
and has eigenvalues $|\Xvec|^2$ and 0. The corresponding (normalized) eigenvectors are then given by
\newcommand{\pvec}{\mathbf{p}}
\begin{equation}
 \pvec_1=\left[\frac{\sgn{X_2}X_1}{|\Xvec|}, \frac{|X_2|}{|\Xvec|}\right]^T
\end{equation}
and
\begin{equation}
 \pvec_2=\left[-\frac{\sgn{X_1}X_2}{|\Xvec|}, \frac{|X_1|}{|\Xvec|}\right]^T.
\end{equation}
Performing now the rotation $\Yvec=\Wh^T\Xvec$ with $\Wh=[\pvec_1, \pvec_2]$ one obtains
\begin{equation}
 \Yvec=[\sgn{X_2}|\Xvec|, 0]^T.
\end{equation}
The fact that the second component of $\Yvec$ is zero regardless of the entries of $\Xvec$ makes it obvious that exactly one half of the information is lost, i.e., $\relLoss{\Xvec\to\Yvec}=\frac{1}{2}$. This also corresponds to the result obtained in Section~\ref{sec:pcaloss} for $N=2$ and $n=1$.
}
{}

\vspace{0.5cm}
\bibliographystyle{IEEEtran}
\bibliography{IEEEabrv,/afs/spsc.tugraz.at/project/IT4SP/1_d/Papers/InformationProcessing.bib,%
/afs/spsc.tugraz.at/project/IT4SP/1_d/Papers/ProbabilityPapers.bib,%
/afs/spsc.tugraz.at/user/bgeiger/includes/textbooks.bib,%
/afs/spsc.tugraz.at/user/bgeiger/includes/myOwn.bib,%
/afs/spsc.tugraz.at/user/bgeiger/includes/UWB.bib,%
/afs/spsc.tugraz.at/project/IT4SP/1_d/Papers/InformationWaves.bib,%
/afs/spsc.tugraz.at/project/IT4SP/1_d/Papers/ITBasics.bib,%
/afs/spsc.tugraz.at/project/IT4SP/1_d/Papers/HMMRate.bib,%
/afs/spsc.tugraz.at/project/IT4SP/1_d/Papers/ITAlgos.bib}

\end{document}